\documentclass[11pt,reqno]{amsart}

\usepackage[hmargin=4cm,vmargin=4cm]{geometry}

\usepackage{amsthm,amssymb}

\usepackage{graphicx}

\usepackage{amscd,amssymb}
\usepackage[leqno]{amsmath}
\usepackage{tikz}
\usepackage{tikz-3dplot}
\usepackage{tikz-cd}
\usepackage{faktor}

\newtheorem  {theorem}                  {Theorem}
\newtheorem* {theorem*}                   {Theorem}

\newtheorem {lemma}[theorem] {Lemma}

\newtheorem* {prop*}     {Proposition}

\newtheorem {Corollary}[theorem]                 {Corollary}

\theoremstyle{definition}
\newtheorem {defi}[theorem] {Definition}
\newtheorem {Remark} [theorem]         {Remark}

\newtheorem* {Example*}    {Example}

\def\Z{\mathbb{Z}}

\newcommand{\pp}[2]{\frac{\partial#1}{\partial#2}}

\title{Hydrodynamic and symbolic models of computation with advice}

\date{}
\author{Robert Cardona}

\address{Robert Cardona, Departament de Matem\`atiques i Inform\`atica, Universitat de Barcelona, Gran Via de Les Corts Catalanes 585, 08007 Barcelona, Spain; Centre de Recerca Matemàtica, Campus de Bellaterra, Edifici C, 08193, Barcelona, Spain.}
\email{robert.cardona@ub.edu}
\thanks{Robert Cardona acknowledges financial support from the Margarita Salas postdoctoral
contract financed by the European Union-NextGenerationEU, as well as from the LabEx IRMIA, the Universit\'e de Strasbourg and Instituto de Ciencias Matemáticas. This work was partially supported by the AEI grant PID2019-103849GB-I00 / AEI / 10.13039/501100011033, AGAUR grant 2017SGR932 and the project Computational,
dynamical and geometrical complexity in fluid dynamics -  AYUDAS FUNDACIÓN
BBVA A PROYECTOS INVESTIGACIÓN CIENTÍFICA 2021.}

\begin{document}

\begin{abstract}
Dynamical systems and physical models defined on idealized continuous phase spaces are known to exhibit non-computable phenomena, examples include the wave equation, recurrent neural networks, or Julia sets in holomorphic dynamics. Inspired by the works of Moore and Siegelmann, we show that ideal fluids, modeled by the Euler equations, are capable of simulating poly-time Turing machines with polynomial advice on compact three-dimensional domains. This is precisely the complexity class $P/poly$ considered by Siegelmann in her study of analog recurrent neural networks. In addition, we introduce a new class of symbolic systems, related to countably piecewise linear transformations of the unit square, that is capable of simulating Turing machines with advice in real-time, contrary to previously known models.
\end{abstract}

\maketitle

\section{Introduction}

Computational aspects of dynamical and physical systems can be studied from a variety of intertwined perspectives such as numerical simulation, computability theory, computational complexity, or analog computation. The last one understands a dynamical system as a computing device that takes an input (the initial condition) and reaches some region of the phase space encoding the output of the process. Combining symbolic dynamics and the Turing machine model, Moore showed in his seminal work \cite{Mo, Mo90} that even low-dimensional dynamical systems are capable of universal computation, thus unveiling the undecidability of some of their properties. Since then, several dynamical systems coming from physical models have been shown to be capable of simulating universal Turing machines. Examples include 3D optical systems~\cite{RTY}, analog recurrent neural networks~\cite{S95}, high dimensional potential wells \cite{T1} and more recently incompressible fluids in various contexts \cite{CMP2,CMP4,CMPP2}.\\

Beyond Turing-computability arises computability with advice: computational models that can compute more than the classical Turing machines (other examples include Turing's oracle machines \cite{Tu}). Dynamical systems modeled in continuous phase spaces allow the presence of real numbers and infinite precision that can lead to non-computable phenomena. Even if those dynamical systems can represent physical models that are highly idealized and hence physically non-realizable, it is interesting from a theoretical of point of view to understand which models can exhibit such computational behavior. For example, the wave equation admits non-computable solutions even if one chooses computable initial data \cite{PR}, see also \cite{WZ, GZ}. Some examples in purely dynamical contexts include the existence of non-computable Julia sets \cite{BY}, and polynomial planar flows with a non-computable number of periodic orbits \cite{GZ2}. This phenomenon is expected to disappear under perturbations in compact spaces, see \cite{BSR}. A complexity class that contains non-computable languages is $P/poly$, which is the set of languages recognized by polynomial-time Turing machines with polynomial advice (see Section \ref{ss:TM} for details). In her influential work \cite{Sig, Sig1, Sig2}, Siegelmann showed that neural networks with real weights can simulate those machines in polynomial time. Other discrete dynamical systems were shown to be computationally equivalent to $P/poly$ by Bournez and Cosnard \cite{BC, Bo}.\\

 In this work, we establish that ideal fluids on three-dimensional geometric domains are also capable of simulating polynomial-time Turing machines with polynomial advice. Recall that given a three-dimensional manifold, with or without boundary, the motion of an ideal fluid (i.e. incompressible and without viscosity) is modeled by the Euler equations
 \begin{equation*}
\begin{cases}
\frac{\partial}{\partial t} u + \nabla_u u &= -\nabla p\,, \\
\operatorname{div}u=0\,,
\end{cases}
\end{equation*}
where $p$ stands for the hydrodynamic pressure and $u$ is the velocity field of the fluid, which is a non-autonomous vector field on $M$ tangent to its boundary. Here $\nabla_u u$ denotes the covariant derivative of $u$ along itself, and $\operatorname{div}$ is the divergence associated with the Riemannian metric $g$. A stationary solution to the Euler equations is an autonomous vector field on $M$ whose integral curves represent the particle paths of a fluid in equilibrium.
\begin{theorem}\label{thm:main}
Given a polynomial-time Turing machine with polynomial advice $(T,a)$, there exists a three-dimensional toroidal domain $U$ equipped with some Riemannian metric $g$, and a stationary solution to the Euler equations in $(U,g)$ that simulates $T$ in polynomial time.
\end{theorem}
There are several ways in which a continuous system can simulate a Turing machine, but generally, it roughly means that each step-by-step computational process of the machine is encoded in the evolution of some orbit of the system. In the statement, we have chosen for simplicity a domain diffeomorphic to a solid torus, however, one can choose as well any closed three-manifold or $\mathbb{R}^3$ for example. To establish Theorem \ref{thm:main}, we first need to introduce a new notion of simulation (see Definition \ref{def:simu}), inspired by the classical ones in the literature \cite{Br,Mo, GCB1, GCB2, BGP}, that provides a well-defined notion of time complexity in the context of conservative ODEs. Then the key step is showing that one can simulate Turing machines with polynomial advice via an analog shift that is mostly injective, which allows us to embed its evolution in a conservative map defined on a finite set of rectangles in a disk (Theorem \ref{thm:blocks}). This requires several technical adaptations from the case of Turing machines without advice as in \cite{Mo, CMPP2}. Finally, we extend the map by rectangles to an area-preserving diffeomorphism of the disk that is then realized as the first-return map of a steady Euler flow on some transverse disk-like section. We conclude by proving that the resulting flow simulates (in polynomial time) the starting Turing machine with advice according to Definition \ref{def:simu}.

We point out that Theorem \ref{thm:main} gives a lower bound on the computational capacity of stationary ideal fluids, and an interesting problem that we do not address here is finding upper bounds for this capacity.\\

In the second part of this paper, we continue the investigation of dynamical systems capable of simulating machines with advice by introducing a new family of symbolic systems that we call ``countable generalized shifts" and that generalize those introduced by Moore \cite{Mo} in a completely different way than analog shift maps \cite{Sig}. We study some dynamical and computational properties of these systems, showing that they can compute $P/poly$. This symbolic model can be partially embedded in the evolution of a countably piecewise linear map of the unit square. This provides an alternative symbolic model for computation with polynomial advice, which in addition has some advantages with respect to previously known dynamical models such as real-time simulation.\\

\textbf{Organization of the paper.} In Section \ref{sec:Symb}, we review the definitions of generalized shift and analog shift map. In Section \ref{sec:Diff}, we prove that there exist area-preserving diffeomorphisms of the disk that simulate in polynomial time Turing machines with advice. Section \ref{s:STflows} starts by introducing time complexity in conservative ODEs by defining an appropriate notion of simulation. It is then shown that there are Euler flows simulating any polynomial-time Turing machine with polynomial advice according to such definition. Finally, in Section \ref{s:CGS} we introduce countable generalized shifts, study some of their dynamical properties and analyze their computational power.\\

\textbf{Acknowledgements:} The author is grateful to Cristopher Moore, whose useful correspondence about transformations of the disk preserving the square Cantor set inspired this work. Thanks to Daniel S. Gra\c{c}a for helpful comments concerning the time complexity of computations in continuous systems, and to Daniel Peralta-Salas for useful discussions.

\section{Symbolic dynamics}\label{sec:Symb}
In this section, we recall several definitions, such as Turing machine with advice, Moore's generalized shifts \cite{Mo} and the analog shift map, a generalization proposed by Siegelmann \cite{Sig}.
\subsection{Turing machines with polynomial advice}\label{ss:TM}
We define a Turing machine $T=(Q,q_0,q_{halt},\Sigma,\delta)$ by the following data:
\begin{itemize}
\item[-] A finite set $Q$ of ``states'' containing two particular (distinct) states: the initial state $q_0 \in Q$ and the halting state $q_{halt} \in Q$.
\item[-] A finite set $\Sigma$ which is the ``alphabet'' and that has cardinality at least two. It contains a specific symbol denoted by $0$ that is also called the ``blank symbol".
\item[-] A transition function $\delta:Q\setminus \{q_{halt}\} \times \Sigma \longrightarrow Q\times \Sigma \times \{-1,0,1\}$.
\end{itemize}

A pair $(q,t)\in Q\times \Sigma^{\mathbb{Z}}$ is a configuration of the machine if the symbols of $t$ are all zero except for finitely many of them. We say in this case that the tape $t$ is compactly supported. Turing machines can be defined as well over tapes with infinitely many non-zero symbols. Hence if $A\subset \Sigma^{\mathbb{Z}}$ denotes the set of sequences that have all but finitely many symbols equal to zero, the space $\mathcal{P}=Q\times A$ is the space of configurations of the machine. When writing a tape $t=(t_i) \in \Sigma^{\mathbb{Z}}$ in the form
$$ ...t_{-1}.t_0t_1... $$
we will use a dot to specify that the position zero lies at the right of the dot.

The evolution of a Turing machine is described as follows. At any given step of the algorithm, we will denote by $q\in Q$ the current state, and by $t=(t_n)_{n\in \mathbb{Z}}\in \Sigma^\mathbb{Z}$ the current tape. Given an input tape $s=(s_n)_{n\in \mathbb{Z}}\in \Sigma^{\mathbb{Z}}$ the machine runs by applying the following algorithm:

\begin{enumerate}
\item We initialize the machine by setting the current state $q$ to be $q_0$ and the current tape $t$ to be the input tape $s$.
\item If the current state is $q_{halt}$ then halt the algorithm and return $t$ as output. Otherwise compute $\delta(q,t_0)=(q',t_0',\varepsilon)$, with $\varepsilon \in \{-1,0,1\}$.
\item Change the symbol $t_0$ by $t_0'$, obtaining the tape $\tilde t=...t_{-1}.t_0't_1...$.
\item Shift $\tilde t$ by $\varepsilon$ obtaining a new tape $t'$. The new configuration is $(q',t')$. Return to step $(2)$. Our convention is that $\varepsilon=1$ (resp. $\varepsilon=-1$) corresponds to the left shift (resp. the right shift).
\end{enumerate}
This algorithm determines a global transition function 
$$\Delta: (Q\setminus \{q_{halt}\}) \times A \longrightarrow \mathcal{P}$$
that sends a configuration to the configuration obtained after applying a step of the algorithm. The global transition function can also be trivially extended to $(Q\setminus \{q_{halt}\}) \times \Sigma^{\mathbb{Z}}$, since a step of the algorithm is well-defined even for non-compactly-supported tapes.\\

 A polynomial-time Turing machine $T$ is a machine that halts for any given input $t_{in}$ of size $n$ in at most $P(n)$ steps, where $P(n)$ is some polynomially-bounded function. For our purposes, we shall follow the simplifying convention in \cite{Sig,Sig2} where we only consider inputs with $t_{k}=0$ for each $k<0$, and say that an input $t_{in}\in \Sigma^{\mathbb{Z}}$ is of size $n$ if it is of the form
$$t_{in}=...0.t_{0}...t_n0...,$$
with $t_n\neq 0$. A polynomial-time machine with polynomial-sized advice $(T,a)$ comes equipped with an infinite collection of strings $a=\{a_n\}_{n\in \mathbb{N}}$ such that $a_n\in \Sigma^{p(n)}$ for some increasing polynomially-bounded function $p(n)$. Given an input $t_{in}$ of size $n$, the machine has access to $a_n$ in one computational step. To make a very concrete dynamical interpretation of this, given $T$ and $\{a_n\}_{n\in \mathbb{N}}$, let us add a new state $\tilde q_0$ in $Q$ which will be an auxiliary initial state, and the machine with input $(\tilde q_0,t_{in})$ applies a preliminary step of the algorithm that sends $(\tilde q_0,t_{in})$ to the configuration $(q_0,t_{in}^a)$ where if $t_{in}=...0.t_0...t_n0...$, the tape $t_{in}^a$ is 
$$ t_{in}^a=...0.t_0...t_na_1^n...a_{p(n)}^n0... $$
where we wrote the advice string as $a_n=a_1^n...a_{p(n)}^n$. For such a Turing machine, the global transition function $\Delta$ can be extended to the set $\{\tilde q_0\}\times \Sigma^{\mathbb Z}$, by declaring that on elements of the form $(\tilde q_0,t_{in})$ we have $\Delta(\tilde q_0,t_{in})=(q_0,t_{in}^a)$, and any arbitrary extension for the other elements in $\{\tilde q_0\} \times \Sigma^\mathbb{Z}$. Throughout the paper, whenever we talk of a Turing machine with advice, we refer to this polynomially-bounded version of it.

\subsection{Generalized shifts and analog shifts}\label{ss:GS}
Let us now review the two classes of symbolic systems that have been related respectively to Turing machines and Turing machines with advice.\\

\textbf{Generalized shifts.}
We recall here the definition of generalized shifts, following the original definition of Moore. Let $A$ be a finite alphabet and $s=(s_i)_{i\in \mathbb{Z}} \in A^\mathbb{Z}$ an infinite sequence. A generalized shift is given by two maps. First, a map
\begin{align*}
G:A^\mathbb{Z} &\longrightarrow A^{\mathbb{Z}},
\end{align*}
together with a finite domain of dependence $D_G=\{i,...,i+l-1\}$ and a finite domain of effect $D_e=\{n,...,n+m-1\}$. These domains indicate that $G(s)$ is equal to $s$ except maybe along the positions in $D_e$, which are modified according to the symbols in the positions $D_G$ of $s$.

Secondly, a map
$$ F:A^{\mathbb{Z}}\longrightarrow \mathbb{Z}. $$
with a finite domain of dependence $D_F=\{j,...,j+r-1\}$, i.e. the image of $F$ only depends on the symbols of the sequence in the positions $D_F$. The image of $s$ by the generalized shift $\phi:A^\Z \rightarrow A^\Z$ is then defined as follows:
\begin{itemize}
\item[-] change $s$ by $G(s)$, this modifies potentially the symbols in positions $D_e$ of $s$,
\item[-] shift $G(s)$ by $F(s)$, and the resulting sequence is by definition $\phi(s)$.
\end{itemize}
As we did for Turing machines, we choose by convention that a positive (negative) value of $F(s)$ means that we shift $|F(s)|$ times to the left (right).

The previous definition is appropriate from a notational point of view to easily define the generalization of analog shift maps. However, generalized shifts admit an alternative simpler notation (as used in \cite{CMPP2}) that we will use as well. \\

\paragraph{\emph{Alternative notation.}} The map $G$, once we have fixed the domains of dependence and effect, is completely determined by assigning to each word in $A^{l}$ a word in $A^{m}$. Hence, by abusing notation, we will understand $G$ as a map 
$$G:A^l \rightarrow A^m,$$ and write $G(t_1...t_l)=t_1'...t_m'$. Similarly, the image by $F$ of $s$ is completely determined by the symbols in positions $D_F$, so abusing notation we can think of $F$ as a map $$F:A^m \rightarrow \mathbb{Z}.$$ Then, given a sequence $s$ the sequence $\phi(s)$ is obtained by replacing the symbols in positions $n,...,n+m-1$ of $s$ by $G(s_i...s_{i+l-1})$ and shifting the resulting sequence by $F(s_n...s_{n+m-1})$. Unless otherwise stated, and without loss of generality, we will assume that a generalized shift is such that $D_e=D_G$. \\

Let us state an obvious property of generalized shifts that we will use later.

\begin{lemma}\label{lem:GS}
Let $\phi$ be a generalized shift. Then for any sequence $s\in A^{\mathbb{Z}}$, its image $\phi(s)$ coincides with $s$ after a shift $-F(s)$ except along the symbols at positions $D_G$.
\end{lemma}

\begin{proof}
The sequence $\phi(s)$ is obtained by changing some of the symbols at positions $D_G$ of $s$ and shifting by  $F(s_i...s_{i+r-1})$. Hence shifting $\phi(s)$ by $-F(s_i...s_{i+r-1})$ yields a sequence that coincides with $s$ except maybe along the symbols in positions in $D_G$.
\end{proof}

\paragraph{\textbf{The analog shift map}}

Analog shift maps \cite{Sig} are defined similarly to generalized shifts, except that the domain of effect of $G$ can be infinite (in one or both directions). Hence in this case an analog shift map $\phi: A^\mathbb{Z} \rightarrow A^\mathbb{Z}$ is specified by 
$$ F: \mathbb{A}^{\mathbb{Z}}\longrightarrow \mathbb{Z}, \enspace D_F=\{i,...,i+r-1\} $$
with a finite domain of dependence $D_F$, and
$$G: A^\mathbb{Z} \rightarrow A^\mathbb{Z}, $$
with a finite domain of dependence $D_G=\{j,...,j+l-1\}$. However, the domain of effect of $G$ depends on the symbols in position $D_G$ of the sequence $s$ and can be a finite number of consecutive integers, a one-sided infinite sequence of consecutive integers or all $\mathbb{Z}$.

\section{Simulation with area-preserving diffeomorphisms}\label{sec:Diff}

In this section, we show how to simulate Turing machines with advice via compactly supported area-preserving diffeomorphisms of the disk. To prove this fact, we will first show how to simulate such machines via a globally injective piecewise linear area-preserving map defined on a finite number of disjoint rectangular domains on a disk. These domains will be Cantor blocks, even though their images by this piece-wise map will not always be Cantor blocks. 

\begin{defi}
The square Cantor set is the product set $C^2:=C\times C \subset I^2$, where $C$ is the (standard) Cantor ternary set in the unit interval $I=[0,1]$. A Cantor block is a block of the form $B=[\frac{a}{3^i},\frac{a+1}{3^i}]\times [\frac{b}{3^j},\frac{b+1}{3^j}] \subset I^2$, where $i,j$ are two non-negative integers and $a<3^i$, $b<3^j$ are two non-negative integers that do not have any $1$ in their ternary expansion.
\end{defi}

We identify sequences $s=(...s_{-1}.s_0s_1...)$ in $\{0,1\}^\mathbb{Z}$ with points $C^2$ via the following bijection.
\begin{align}\label{eq:ide}
\begin{split}
e:\{0,1\}^{\mathbb{Z}} &\longrightarrow C^2 \\
 (s_i)_{i\in \mathbb{Z}} &\longmapsto \left(\sum_{i=1}^\infty s_{-i}\frac{2}{3^i}, \sum_{i=0}^\infty s_i\frac{2}{3^{i+1}}\right)
 \end{split}
\end{align}

If we choose an alphabet with $k$ symbols instead of only two, we can encode all the sequences of the alphabet in a square Cantor set ${C_k}^2$ where $C_k$ is a Cantor set of the interval obtained by iteratively removing $k-1$ open subintervals. Then Cantor blocks are defined analogously. In the statement below $D$ denotes some disk of big enough radius that contains the unit square $I^2$.

\begin{theorem}\label{thm:blocks}
Given a Turing machine with advice $(T,a)$, there exist two families  $\mathcal{B}=B_1,...,B_{N-1}, B_N$ and  $\mathcal{B}'=B_1',...,B_{N-1}'$ of pairwise disjoint Cantor blocks, a rectangular domain $B_N'\subset D$ disjoint from $\mathcal{B}'$, and an injective piecewise linear area-preserving map $F: \sqcup B_i \rightarrow B_i'$ such that $F|_{C^2\cap \mathcal{B}}$ simulates $(T,a)$ in polynomial time.
\end{theorem}
A concrete property that formalizes simulation in polynomial time in this context is that there exists a computable map $\varphi$ that maps the configurations of $(T,a)$ to points on a square Cantor set such that if $(T,a)$ halts with input $(q_0,t_{in})$ and output $(q_{halt},t_{out})$ in $k$ steps, then the orbit of $F$ through $\varphi(q_0,t_{in})$ reaches $\varphi(q_{halt},t_{out})$ in $q(k)$ iterations of $F$, where $q$ is some polynomially-bounded function. 

\begin{proof}
 Following \cite{Sig2}, we recall how the analog shift map is capable of simulating, in polynomial time, a given polynomial-time Turing machine with polynomial advice. This will help us understand the technical adaptions that we need to do to obtain partly injective dynamics, which is key to obtaining a map $F$ that is injective.
  
Given a Turing machine with advice $(\hat T,a)$ with states $Q$ and alphabet $\Sigma$, we first forget the ``additional" state $\tilde q_0$ that plays the role of initial state and is used to include the advice on the inputs via $\Delta(\tilde q_0, t_{in})= (q_0, t_{in}^a)$ as described in Section \ref{ss:TM}. We will denote this simpler machine by $T$, it does not have $\tilde q_0$ in its set of states and $q_0$ is its initial state instead. To be able to simulate $\hat T$ by an analog shift, the idea is to add all of the advice, i.e. all the strings concatenated in a one-sided infinite sequence $a_{\infty}=...a_2a_1a_0$, to any input $t_{in}$ in one step. Then, the dynamics are constructed to ``extract" the string corresponding to the input in polynomial time. To explain what we mean by this, first assume that we can reach a configuration of the form $$a_\infty.t_0...t_n0...,$$
that is, all of the advice has been added to an input. Then we modify $T$ so that it changes, in polynomially many steps, this tape (with a suitable state) to a tape of the following form.
$$...a_{n+2}a_n.t_0...t_na_n0...$$ 
This is explained in \cite[Chapter 12]{Sig2}: there exists a modified Turing machine $\tilde T$ that works as $T$ with a preliminary extraction of the advice. The set of states and symbols of $\tilde T$ contains those of $T$ but has as well an extra marker symbol $d$ in its alphabet and an extra set of states $\tilde q_1,...,\tilde q_r$. The initial state of $\tilde T$ is $\tilde q_1$ (instead of $q_0\in Q)$ and it simulates $T$ on inputs of the form $\tilde t=(a_\infty.t_0...t_n0...)$ in polynomial time as follows. First, in $g(n)$ steps (where $g$ is some polynomially-bounded function) the configuration $(\tilde q_1, \tilde t)$ reaches the configuration $(\tilde q_r, t')$ with
$$t'=(...a_{n+1}d.t_0...t_na_n0...), $$
and from there simulates $T$ in polynomial time. This is done step by step as $T$ would do it, but moving the marker $d$ to the left when necessary to keep track of the relevant part of the tape (everything kept at the left of $d$ is not used for the computation). To keep the notation simple, let $Q$ and $\Sigma$ denote the set of states and the alphabet of $\tilde T$. \\

We will now show how to simulate $\tilde T$ and a preliminary step that includes $a_\infty$ in the left side of the tape of any input by an analog shift map. Consider the alphabet $\tilde A=Q\sqcup \Sigma \sqcup \{\tilde q_0\}$, and we identify the configurations of $\tilde T$ to $A^\mathbb{Z}$ by the map 
\begin{align*}
\varphi : Q\times \Sigma^{\mathbb{Z}} &\longrightarrow {\tilde A}^\mathbb{Z}\\
		(q,t) &\longmapsto (...t_{-1}.qt_0t_1...).
\end{align*}
Since $\tilde T$ is a standard Turing machine (no advice is being added for the moment since we are not using $\tilde q_0$), using \cite[Theorem 7]{Mo}, there exists a generalized shift $\phi_{GS}$ with $D_F=D_G=D_e=\{-1,0,1\}$ defined on the alphabet $A=\tilde A\setminus \{\tilde q_0\}$ such that the global transition function of $\tilde T$ is semi-conjugate to $\phi_{GS}$ by this identification. We extend this generalized shift to an analog shift map $\phi:\tilde A^\mathbb{Z}\rightarrow \tilde A^{\mathbb{Z}}$ by imposing:
\begin{equation}\label{eq:inf}
\begin{cases}
G(0.\tilde q_0t)&=(a_\infty.\tilde q_1t)\enspace \text{for any } t\in \Sigma,\\
F(0.\tilde q_0 t)&=0 \enspace \text{for any } t\in \Sigma.
\end{cases}
\end{equation}
The dot in the words is used to specify that $D_G$ is still $\{-1,0,1\}$, but the domain of effect of $G$ for the words of the form $(0.\tilde q_0t)$ is now every integer smaller or equal to $1$. For each $(s_{-1}.s_0s_1)$ where $\phi_{GS}$ was already defined, the map $G$ still only changes the symbols in position $-1,0,1$. \\

Let us check that $\phi$ simulates $\hat T$ in polynomial time, and that $\phi$ can be assumed to be injective in $\varphi(Q\setminus \{q_{halt}\}\times \Sigma^{\mathbb{Z}})$. We encode an input $(...0.t_0...t_n0...)$ in $\tilde A^\mathbb{Z}$ as the sequence $(...0.\tilde q_0t_0...t_n0...)$. After one step of $\phi$, we reach $(a_{\infty}.\tilde q_1t_0...t_n0...)$, and then proceed simulating $\tilde T$. This shows that $\phi$ simulates $\hat T$ in polynomial time. If we want $\phi$ to be injective, we need to assume that $\tilde T$ is reversible. This is possible for the following reasons. A classical method of Benett \cite{Ben} shows how to simulate $T$ by some reversible Turing machine with three tapes. The simulation is polynomial in time, in fact, if $\tilde T$ halts in $K$ steps with some input,  then the reversible machine halts in $O(K^{1+\varepsilon})$ steps.  A reversible three-tape Turing machine can be simulated by a reversible one-tape Turing machine, such simulation is also polynomial in time (in this case the simulation is quadratic see e.g. \cite[Proposition 1]{Ax}). We can also assume that the initial state of $\tilde T$, which is $\tilde q_1$, is not in the image of the transition function of $\tilde T$ as discussed in \cite[Section 6.1.2]{Mor}. Then $\tilde \Delta: (Q\setminus \{q_{halt}\}) \times \Sigma^{\mathbb{Z}}\rightarrow Q\times \Sigma^{\mathbb{Z}}$, the global transition function of $\tilde T$, is injective, and so will be $\phi$.\\

We are now ready to construct a map by blocks that encodes the dynamics of $\phi$. First, the symbolic dynamics where no advice is added: an application of \cite[Lemma 0]{Mo} implies that $\phi|_{A^{\mathbb{Z}}}=\phi_{GS}$ is induced by a piecewise area-preserving linear map defined on blocks $B_1,...,B_{N-1}$ of some square Cantor set $\tilde C$ (associated to an alphabet of $|\tilde A|$ symbols). This set of blocks never intersects in the block determined by fixing the zero symbol of a sequence to be $\tilde q_0$, since $\phi_{GS}$ is defined on $A^\mathbb{Z}\subset \tilde A^\mathbb{Z}$. Each of these blocks is contained in the block determined by an element $(s_{-1},s_0,s_1)\in A^3$ of the domain of dependence. The fact that $\tilde T$ is reversible implies that the image blocks $B_1',...,B_{N-1}'$ of $B_1,...,B_N$ are also pairwise disjoint. The linear area-preserving map 
$$F:\bigsqcup_{i=1}^{N-1} B_i \rightarrow \bigsqcup_{i=1}^{N-1} B_i'$$
given by Moore's lemma induces on $\tilde C$ the generalized shift $\phi_{GS}$ via an encoding of the form of \eqref{eq:ide}. No block $B_i, B_i'$ contains a sequence whose symbol in position zero is $\tilde q_0$ by construction, since $\tilde q_0$ is not in the alphabet of $\phi_{GS}$. 

We will now extend $F$ to another rectangular domain so that the map also includes the instructions \eqref{eq:inf} as well. We will first define this rectangle and its image. The rectangle, which will be denoted by $B_N$, is given by the Cantor block of $\tilde C$ corresponding to those sequences that have at position zero the symbol $\tilde q_0$ (this is a horizontal rectangle in $I^2$). Let $\tilde B$ be the block given by those sequences that have at position zero the symbol $\tilde q_1$, this is another horizontal rectangle in $I^2$. No block $B_i'$ intersects $\tilde B$, since $\tilde q_1$ is not in the image of the transition function of $\tilde T$.  Now we define the image block $B_N'$ as the image of $\tilde B$ by the translation 
$$ \tau(x,y):=(x+\alpha,y),$$
where $\alpha$ is the number associate do the left-infinite sequence $(a_\infty.0...)$ via a map like \eqref{eq:ide}. That is, if the alphabet was of two symbols and $a_\infty=(...s_{-2}s_{-1}.0...)$, then $\alpha=\sum_{i=1}^\infty s_{-i}\frac{2}{3^{-i}}$.

It is clear that $B_N'$ is disjoint from $B_1',...,B_{N-1}'$, since we just slightly translated horizontally the horizontal Cantor block $\tilde B$ and thus the part of $B_N'$ that is not in $\tilde B$ is not even in $I^2$, so it cannot intersect any other Cantor block. We define $F|_{B_N}$ to be a composition of the trivial vertical translation of $B_N$ into $B_N'$ and $\tau$. Observe that if we start with a sequence $(...0.\tilde q_0 t_0...t_n0...)$ in $\tilde A^\mathbb{Z}$, it is encoded to $B_N\subset \tilde C$, and applying $F$ it is sent to the point whose associated sequence is $(a_\infty.\tilde q_1 t_0...t_n0)$. In other words, given an input $(...0.\tilde q_0t_0...t_n0...)$, in one step all the advice of the machine is included in the sequence by $F$, and then the next iterations of $F$ will follow the computations of $\tilde T$ as specified by $\phi_{GS}$. We conclude that $F:\bigsqcup_{i=1}^N B_i \rightarrow \bigsqcup_{i=1}^N B_i'$ induces on $\tilde C$ a map that simulates $T$ in polynomial time.
\end{proof}
Observe that even if the initial machine $T$ is already reversible, the proof of Theorem \ref{thm:blocks} shows that the simulation is only done in polynomial time, not in real-time.  This is because one needs to preprocess $a_\infty$ by introducing the auxiliary machine $\tilde T$.

The map $F$ extends to a compactly supported area-preserving diffeomorphism of the disk arguing exactly as in \cite[Proposition 5.1]{CMPP2}.

\begin{Corollary}\label{coro:area}
Given a polynomial-time Turing machine with polynomial advice $(T,a)$, there exists a compactly supported area-preserving diffeomorphism of a disk $D$ that simulates $T$ in polynomial time.
\end{Corollary}

The simulation is described as follows. The symbolic system $\phi:A^\mathbb{Z}\rightarrow A^\mathbb{Z}$ (that simulates $(T,a)$ in polynomial time) that is constructed in the proof of Theorem \ref{thm:blocks} is semiconjugate, along the sequences that encode configurations of $(T,a)$, to the area-preserving diffeomorphism of the disk.

Regarding computability, the encoding \eqref{eq:ide} is of course computable on compactly supported sequences, and so is the map by blocks except along the block that simulates the first step of the machine (block $B_N$ in the proof of Theorem \ref{thm:blocks}). This is necessary since the dynamics applied to that block adds to the computational process the possibly non-computable advice string.

\section{Computation with advice in ideal fluid motion}\label{s:STflows}

In this section, we introduce a definition of simulation of Turing machines by volume-preserving autonomous flows that admits a well-defined notion of time complexity. We show that for a given Turing machine with advice there exists a stationary solution to the Euler equations, on a compact three-dimensional geometric domain equipped with some Riemannian metric, that simulates it.

\subsection{Simulation and complexity with conservative ODEs}

We are interested in showing that a continuous dynamical system simulates a polynomial-time Turing machine, we thus need to specify how to measure ``computational time" (steps of the computation) in a continuous-time model. Throughout this section, every object we consider, like a manifold or a vector field, will be assumed to be smooth. Let $X$ be an autonomous vector field on a manifold $M$ (for concreteness, one might think of an ordinary differential equation defined in $\mathbb{R}^n$).  Given a point $p\in M$ the integral curve of $X$ through $p$ is the solution to the system
\begin{equation}\label{eq:ODE}
\begin{cases}
y'(t)= X(y(t)),\\
y(0)=p.
\end{cases}
\end{equation}
There are several ways to define continuous-time simulation of a Turing machine $T$ by a vector field $X$. One approach, introduced in \cite{T1} and followed in \cite{CMPP2, CMP2, CMP4}, is to define simulation by requiring that the halting (perhaps with a prescribed finite part of the output) of the machine for a given input is equivalent to the integral curve of $X$ through a computable point $p\in M$ intersecting a computable open set $U\subset M$. This definition does not impose a ``step-by-step" simulation, as it is more customary in previous works \cite{Br, GCB1, GCB2, KM99}, although the constructions in \cite{T1, CMPP2,CMP2,CMP4} do in fact provide step-by-step simulations. Inspired by \cite{GCB1,GCB2}, we consider the following definition of simulation that takes into account the behavior of the time parameter.
\begin{defi}\label{def:simu}
Let $X$ be a vector field on a manifold $M$. Let $T$ be a Turing machine and denote by $\Delta$ its global transition function, and by $\varphi: \mathcal{P}\hookrightarrow M$ an injective encoding of the configurations of $T$ in $M$. We say that $X$ simulates $T$ if for each initial configuration $(q_0,s)\in \mathcal{P}$, there is a constant $K_{s}$ such that the solution $y(t)$ to $y'(t)=X(y(t))$ with initial condition $y(0)=\varphi(q_0,s)$ satisfies 
\begin{equation}\label{eq:simu}
y(K_{s}\cdot n)=\varphi(\Delta^n(q_0,s)),
\end{equation}
for each $n \leq N$, where $N \in \mathbb{N}\sqcup \{+\infty\}$ is the halting time of $T$ with input $(q_0,s)$. 
\end{defi}
With this definition, the values of $t$ that are integer multiples of $K_{s}$ measure the steps of the algorithm. So given one computation (an initial configuration), each step of the algorithm is performed in the same amount of continuous time. We will then see why this definition behaves well for conservative autonomous flows.

\begin{Remark}\label{rem:polysim}
Definition \ref{def:simu} defines real-time simulation, i.e. one step of the machine corresponds to one step (measured by time multiples of $K_s$) of the system. Polynomial-time simulation can be defined analogously, by replacing $y(K_{s}\cdot n)$ by $y(K_{s}\cdot Q(n))$ in Equation \eqref{eq:simu} for some polynomially-bounded function $Q(n)$.
\end{Remark}
\begin{Remark}\label{rem:openset}
We can further require that the positive trajectory $y(t)$ as above intersects some open set $U_{s}$ (for example, a $\varepsilon$-neighborhood of the image by $\varphi$ of the halting configurations) if and only if $T$ halts with input $(q_0,s)$. This is interesting because the halting of $T$ with a given input is equivalent to an integral curve reaching an explicit open set. In addition, we can require that the orbit $y(t)$ either intersects $U_s$ or remains at a positive distance (bounded from below) from it. 
\end{Remark}
\begin{Remark}\label{rem:setencoding}
It is also possible to use an encoding $\varphi$ that assigns an open set to each configuration (e.g. as in \cite{CMP4}), and then we require that $y(K_{t_i}\cdot n)$ lies in the open set that encodes $\Delta^n(q_0,t_{in})$. Definition \ref{def:simu} readily generalizes to robust simulations as in \cite{GCB1}.
\end{Remark}

Given an autonomous vector field simulating a Turing machine as in Definition \ref{def:simu}, one is tempted to use the time-parameter $t$ of an integral curve as a measure of time complexity. According to the definition of step-by-step simulation, the solution at time $t=k$ corresponds to the $k^{th}$ step of the algorithm. However, as classically treated in the literature \cite{BGP,BP},  the parameter $t$ is not a well-defined measure of time for a very simple reason. One can rescale the vector field $X$ defined in $M$ by considering $\tilde X=fX$ for some positive function $f\in C^\infty(M)$ so that the orbit can compute faster or slower depending on the step of the algorithm that is being simulated. The same happens if we consider a vector field simulating a Turing machine as in Definition \ref{def:simu}, taking as the time step the values of the continuous-time parameter given by integer multiples of $K_s$. However, if $X$ is assumed to preserve some volume form $\mu \in \Omega^n(M)$ (where $n=\dim M$), the following proposition shows that there is a well-defined notion of time complexity. It captures the fact that in the conservative context, only linear speedup is possible. This can be heuristically intuited since non-linear speed-up along an orbit necessarily expands or contracts a small cube in flow-box coordinates near a point of that orbit.

\begin{lemma}
Let $X$ be a vector field on $M$ preserving some volume form $\mu\in \Omega^n(M)$ and simulating a Turing machine $T$ as in Definition \ref{def:simu}. Let $\widetilde X= f\cdot X$, with $f\in C^\infty(M)$ a positive function, be a reparametrization of $X$ that also preserves $\mu$. Then $\widetilde X$ also satisfies Definition \ref{def:simu} with the same encoding and constants $\widetilde K_s= \frac{K_s}{f(\varphi(q_0,s))}$.
\end{lemma}
\begin{proof}
Assume that $X$ satisfies Definition \ref{def:simu} with encoding $\varphi$ and constant $K_s$ for a given initial configuration $(q_0,s)$. Let $\widetilde X= f\cdot X$ be a reparametrization of the vector field $X$ that preserves $\mu$, for some positive function $f\in C^\infty(M)$. This implies that $\mathcal{L}_{\tilde X}\mu= 0$, and is equivalent to
$$ d\iota_{fX}\mu= df\wedge \iota_X\mu=0. $$
Contracting this equation with $X$ again, we see that a necessary and sufficient condition is that $\iota_Xdf=0$, i.e. that $f$ is a first integral of $X$. Let us show that $\tilde X$ satisfies Definition \ref{def:simu} with the same encoding as $X$. Let $(q_0,s)$ be an initial configuration of the machine $T$. Let $\tilde y(t)$ be the integral curve of $\tilde X$ with initial condition $\varphi(q_0,s)$, hence the solution to the system
\begin{equation}\label{eq:repODE}
\begin{cases}
\tilde y'(t)= f\cdot X(\tilde y(t)),\\
\tilde y(0)= \varphi(q_0,s).
\end{cases}
\end{equation}
Since $f$ is a first integral of $X$, it is constant along any orbit of $X$, so we might replace $f$ by the constant $C_s=f(\varphi(q_0,s))$ in Equation \eqref{eq:repODE}. Let $y(t)$ denote the solution to the system \eqref{eq:ODE}. It is standard that $\tilde y(t)$ satisfies
$$ \tilde y(t)= y(C_s\cdot t).$$
Consider the constant $\tilde K_s=\frac{K_s}{C_s}$. Using that $X$ simulates $T$ with constants $K_s$ and encoding $\varphi$, we deduce that
$$ \tilde y(\tilde K_s \cdot n)=y (K_s \cdot n)=\varphi(\Delta^n(q_0,s)),$$
for each $n\leq N$, where $N$ is the halting time of $T$ with input $(q_0,s)$.
\end{proof}
It is clear that if a vector field $X$ satisfies the property described in Remark \ref{rem:openset}, any reparametrization also satisfies it. The time complexity of a computation can be measured in a well-defined way by using the values of $t$ that are integer multiples of $K_s$. From a physical point of view, the value of $f$ along an orbit of $X$ measures the norm of $X$ along that orbit. It is reasonable that computations along an orbit where $X$ has a greater norm (though of as a measure of the ``energy" of the system along that orbit) occur faster in terms of the continuous measure of time $t$. However, the time complexity measured discretely is invariant under these reparametrizations.

\subsection{Stationary Euler flows computing $P/poly$}
Having introduced a well-defined notion of time complexity for conservative vector fields, we will prove in this subsection that given any polynomial-time Turing machine with polynomial advice, there exists a solution to the stationary Euler equations in some compact Riemannian three-manifold that simulates it (polynomially in time) according to Definition \ref{def:simu}.\\

The Euler equations model the dynamics of an ideal (incompressible and without viscosity) fluid on a Riemannian manifold $(M,g)$ and they take the form
\begin{equation*}
\begin{cases}
\partial_t u +\nabla_u.u=-\nabla p,\\
\operatorname{div}u=0.
\end{cases}
\end{equation*}
Here $u$ is the velocity field of the fluid, the scalar function $p$ is the pressure function, and all the differential operators are defined with the ambient metric $g$. A stationary solution is a solution satisfying $\partial_t u=0$, it is hence a time-independent vector field whose integral curves define the particle paths of the fluid. The second equation ensures that $u$ is always a volume-preserving vector field with respect to the Riemannian volume. To prove that there exist stationary solutions that simulate polynomial-time Turing machines with polynomial advice, our main tool will be the connection between Euler flows and Reeb flows in contact geometry established by Etnyre and Ghrist \cite{EG}. This connection was used in \cite{CMPP2} to prove that there exist Turing complete steady Euler flows in three-dimensional compact manifolds. 

Let us recall that on a three-dimensional manifold $M$, a (cooriented) contact structure is a plane distribution $\xi$ defined as the kernel of a one-form $\alpha\in \Omega^1(M)$ that satisfies the non-integrability condition $\alpha\wedge d\alpha\neq 0$. We call $\alpha$ a contact form, and note that any positive multiple $\gamma=f\cdot \alpha$ with $f\in C^\infty(M)$ is another contact form defining $\xi$. Each contact form $\gamma$ uniquely defines a vector field $R$ called the Reeb vector field, which is determined by the equations
\begin{equation*}
\begin{cases}
\gamma(R)=1,\\
\iota_Rd\gamma=0.
\end{cases}
\end{equation*}
Reeb fields will play a role in the proof of Theorem \ref{thm:main}.

\begin{theorem}\label{thm:Eulerpoly}
Let $(T,a)$ be a polynomial-time Turing machine with polynomial advice. There exists a metric $g$ in $S^3$ and a stationary solution to the Euler equations $X$ in $(S^3,g)$ that simulates $T$ in polynomial time.
\end{theorem}
The simulation will be according to Definition \ref{def:simu} and Remark \ref{rem:polysim}.
\begin{proof}
Let $(T,a)$ be a polynomial-time Turing machine with polynomial advice.  By Corollary \ref{coro:area} there exists a compactly supported area-preserving diffeomorphism of a disk
$$ H:D\longrightarrow D $$
that simulates $(T,a)$ in polynomial-time. Concretely, as done in Theorem \ref{thm:blocks}, there exists a symbolic system $\phi:A^\mathbb{Z} \rightarrow A^\mathbb{Z}$ and a computable map $E:\mathcal{P}\rightarrow A^\mathbb{Z}$ encoding the configurations of the machine such that $\phi$ simulates $(T,a)$ in polynomial time. Then $H$ satisfies  $H(e(s))=e(\phi(s))$ for every $s\in E(\mathcal{P})\subset A^\mathbb{Z}$, where $e$ denotes an encoding of the form \eqref{eq:ide} (perhaps using an expansion in base $k$ instead of two) into some square Cantor set $\tilde C^2$.

Fix the contact manifold $(S^3, \xi_{std})$, where $S^3$ is the three-sphere and $\xi_{std}$ the standard tight contact structure. By \cite[Theorem 3.1]{CMPP2}, there exists a contact form $\alpha$ whose Reeb field $R$ exhibits a Poincar\'e disk-like section $D_M \subset M$ whose first-return map is conjugate to $H$. This means that $D_M$ is an embedded disk transverse to the flow, and that there exists a smooth function $\tau:D_M \rightarrow \mathbb{R}$ such that the flow of $R$, that we denote by $\varphi_t:M\longrightarrow M$,  satisfies
$$ \varphi_{\tau(p)}(p)= \psi \circ H \circ \psi^{-1}(p), \text{ for } p\in D_M $$
where $\psi:D\longrightarrow D_M$ is a parametrization identifying $D$ with the disk-like section $D_M$.  

We will now choose a suitable positive function $h\in C^\infty(M)$ so that the first-return time of the flow of $X=h\cdot R$ at $D_M$ is constant and equal to one. First, up to multiplying $R$ by a small enough constant,  we can assume that $\tau(p)<1$ for all $p\in D_M$.  Choose flow-box coordinates $(x,y,z)$ of $U=\{\varphi_z(D')\enspace | \enspace z\in [-\varepsilon,\varepsilon]\} \cong D^2\times [-\varepsilon,\varepsilon]$, where $D'$ is a slightly larger disk-like section containing $D_M$. Denote by $F$ the first-return map on $D'$, which satisfies $F|_{D_M}=\varphi_{\tau(p)}(p)$.  In these coordinates, we have $R=\pp{}{z}$, and the integral curve with initial condition $(x,y,-\varepsilon)$ takes exactly time $\varepsilon$ to hit $D_M$.\\

Construct a smooth function $g:D'\times [-\varepsilon,\varepsilon] \rightarrow \mathbb{R}$ equal to $1$ near the boundary of $D'\times [-\varepsilon,\varepsilon]$ and such that 
$$\int_{-\varepsilon}^0\frac{1}{1+g(x,y,z)}dz=1-\tau(F^{-1}(x,y))+\varepsilon, \text{ for } (x,y) \in D_M.$$ It is clear that such a function exists: for a fixed point $(x,y)$, we are choosing a function depending on $z$ with a given value of the integral above. Varying smoothly the value of the integral we can smoothly vary the function with respect to $z$, parametrically with respect to the two parameters $x$ and $y$. 

We extend $g$ to a smooth function $h$ in $M$ equal to $1$ away from $D'\times [-\varepsilon, \varepsilon]$. Consider the vector field $X=h\cdot R$. We claim that the first-return time of $X$ to $D_M$ is constant and equal to one. Given $p\in D_M$, the solution to the ODE defined by $X$ and initial condition $p$ hits $D_M\times \{-\varepsilon\}$ at a point $(x,y,-\varepsilon)$ after time $\tau(p)-\varepsilon$, since $X=R$ along the piece of orbit outside of $D'\times [-\varepsilon,0)$. In particular, we have $(x,y)=F(p)$. On the other hand, the solution $u(t)$ to the ODE defined by $X$ and initial condition $(x,y,-\varepsilon)$ satisfies 
$$ t=\int_{-\varepsilon}^z \frac{1}{1+g(x,y,z)}dz. $$
It follows that the solution intersects $D_M=\{z=0\}$ when $t=1-\tau(F^{-1}(x,y))=1-\tau(p)$. Hence the time that the flow of $X$ takes to send a point $p$ back to $D_M$ is 
$$\tilde \tau(p)=\tau(p)-\varepsilon+1-\tau(p)+\varepsilon=1$$
 as claimed.\\

We have thus constructed a reparametrized Reeb field $X=h\cdot R$ that has a disk-like Poincar\'e section with first-return time equal to one and conjugated to $H$. It was proved in \cite{EG} that any such vector field is a stationary solution to the Euler equations for some Riemannian metric in the ambient manifold. It only remains to check that $X$ does simulate the Turing machine $T$ according to Definition \ref{def:simu}. Given a configuration $c\in \mathcal{P}$ of the machine, it is mapped to an element of $A^{\mathbb{Z}}$ by some map $E:\mathcal P \longrightarrow A^{\mathbb{Z}}$. The set of sequences $A^\mathbb{Z}$ in the image of $E$ is injectively mapped to the square Cantor set $\tilde C^2$ by the map $e$. As encoding, we choose $\tilde e= \psi\circ e \circ E$ and as constants we take $K_s=1$ for all $s$. The first-return map $F_X$ of $X$ at any point $p\in D_M$ is given by the flow of $X$ at time $1$. Hence, given an initial configuration $(q_0,s)$ of the machine $T$, we consider the solution $y(t)$ to the ODE defined by $X$ and initial condition $\tilde e(q_0,s)$. Using that $F_X=\psi\circ H \circ \psi^{-1}(p)$ and that $y(k)=F_X^k(p)$, we deduce that 
$$ y(Q(n))=\tilde e(\Delta^n(q_0,s)), $$
for each $n$ smaller than the halting time of $T$ with input $(q_0,s)$, where $Q(n)$ is a polynomially-bounded function that comes from the polynomial-time simulation of $(T,a)$ by $H$. This concludes the proof that $X$ simulates $T$ according to Definition \ref{def:simu}, polynomially in time as described in Remark \ref{rem:polysim}.
\end{proof}
We point out that, as mentioned in the context of neural networks \cite{Sig}, a given computation requires only polynomial time with respect to the size of the input and hence is simulated by a finite portion of the associated integral curve. Hence, a given computation is robust to perturbations of a size that depends on the size of the input. Indeed, only finitely many positions of the sequence need to be read to simulate a finite number of iterations of the symbolic system in $\tilde C^2$.
\begin{Remark}
The reparametrization argument used in the proof of Theorem \ref{thm:Eulerpoly} can be applied to the Turing complete Reeb flows constructed in \cite{CMPP2}. This yields a stationary solution to the Euler equations in some Riemannian three-sphere that simulates a universal Turing machine according both to the definition used in \cite{CMPP2} and to Definition \ref{def:simu} that takes into account time complexity.
\end{Remark}

We point out that Theorem 3.1 in \cite{CMPP2} can be applied to any fixed closed contact three-manifold $(M,\xi)$, or open contact three-manifold such as $(\mathbb{R}^3,\xi_{std})$, so one can replace a toroidal domain by closed three-manifold. In addition, the strong property described in Remark \ref{rem:openset} is also true, namely that there exists an open set $U$ such that for any initial configuration $(q_0,s)$, the orbit of $X$ through the explicit point $p_s\in M$ associated with $(q_0,s)$ intersects $U$ if and only if $T$ halts with input $(q_0,s)$. The orbit either intersects $U$ or stays at a positive distance from $U$ uniformly bounded from below. Even if a polynomial-time Turing machine halts in every input, the last property is not meaningless. One could, for example, modify the Turing machine so that halting only occurs on accepted inputs. Then the previous property would mean that the trajectory associated to an input that is accepted intersects some domain $U$, but if the input is not accepted then the associated trajectory remains at a positive distance of $U$. This is seen as follows. 

Let $y(t)$ be the trajectory associated with an initial configuration $(q_0,s)$. All the points representing a halting configuration are encoded in a finite collection of blocks of the square Cantor set contained in $D_M$, and no non-halting configuration is encoded there. Let $U_D$ be a small enough neighborhood of those blocks not intersecting any other block, and $U$ be defined as 
$$ U= \{ \phi^X_t(p)\enspace | \enspace p\in U_D, t\in (-\varepsilon,\varepsilon)\},$$
where $\varphi_t^X:M\longrightarrow M$ denotes the flow defined by $X$. It is clear that $y(t)$ intersects $U$ if and only if $T$ halts with input $(q_0,s)$. The orbit of the initial configuration by the first-return map will always remain in the blocks of the Cantor set associated with non-halting configurations. These blocks are at a distance greater than $\delta$, for some $\delta>0$, of the blocks containing the halting configurations. This shows that the orbit of a non-halting initial configuration stays at a positive distance (bounded from above) of $U$.

\subsection{Variations of the model}\label{s:omod}

We finish this section by discussing other models, either of complexity or of ideal fluids, that could be considered in this context.\\

\paragraph{\textbf{Time complexity via orbit length}}
In \cite{BGP}, another possible measure of the time complexity of simulations with ODEs was proposed. When there is a metric in the ambient space (for example, the Euclidean metric for an ODE defined in $\mathbb{R}^m$), a measure of time that is invariant with respect to reparametrizations is the length of the orbit with respect to the Riemannian metric. This approach is also reasonable in the context of hydrodynamics since the space is endowed with a natural Riemannian metric, the one for which the vector field solves the Euler equations. This point of view can also be taken in our construction in Theorem \ref{thm:Eulerpoly}. Indeed, the vector field $X$ has no zeroes, and the ambient manifold is compact, hence there are constants $c,C\in \mathbb{R}$ such that 
\begin{equation}\label{eq:norm}
 c< g(X,X) < C, 
\end{equation}
where $g$ denotes the metric for which $X$ is a stationary solution to the Euler equations. In particular, the length of an injective piece of an integral curve grows linearly with time. For the flow $X$, the computational steps are given by $K_s=1$ (as in Definition \ref{def:simu}). For a given input of size $n$ the machine halts after $P(n)$ steps (where $P$ is a polynomially-bounded function and $n$ is the size of the input). By Equation \eqref{eq:norm}, the length of the curve up to time $P(n)$ is polynomial as well, so polynomial complexity is well-defined using the approach proposed in \cite{BGP}. Note that it is also possible to construct a metric $\tilde g$ for which $X$ is a stationary solution to the Euler equations that has a norm everywhere equal to one. In that case, the length coincides with the time. The construction of $\tilde g$ is done using the arguments explained in \cite[Section 1.3.4 page 85]{ThC}, by considering a one-form $\alpha$ that is not any more of contact type everywhere, but is instead closed in the solid torus where it satisfies $\alpha(X)=1$.\\

\paragraph{\textbf{Other hydrodynamical systems}}

Ideal fluid flows capable of universal computation have been constructed in other situations, besides stationary flows on geometric three-dimensional domains endowed with an adapted (not fixed a priori) Riemannian metric. Indeed, a natural requirement is to impose that the metric is a fixed natural one, such as the flat metric on the three-torus or the Euclidean metric in $\mathbb{R}^3$. In \cite{CMP4}, it was shown that at the high cost of losing compactness, one can construct stationary solutions to the Euler equation in $\mathbb{R}^3$ with the Euclidean metric that can simulate a universal Turing machine. One can check that the simulation is not as good as the one defined in Definition \ref{def:simu} because it has an exponential slow-down in terms of the continuous-time parameter of the ODE. Even if we use the orbit-length approach to time complexity, one cannot simulate polynomial-time Turing machines in polynomial time using the construction done in \cite{CMP4}. A natural question is then whether, for an arbitrary Turing machine (with or without advice), there exist stationary solutions to the Euler equations in Euclidean space that are capable of simulating it in polynomial time, according to some natural definition of simulation and time complexity.

Similarly, it was proved in \cite{CMP2} that there are time-dependent solutions to the Euler equations in some high enough dimensional closed manifold that are capable of simulating a universal Turing machine. Those solutions not only have an exponential slow-down as well but also rely on constructions of polynomial ODEs \cite{GCB2} that simulate any Turing machine which might not hold for Turing machines with advice. Hence another question is if there are time-dependent solutions to the Euler equations modeling computation with advice. Perhaps this can be shown as well by doing a construction that can use the embedding results in \cite{TdL, T2}, as done in \cite{CMP2}. All these questions can be asked too for viscous fluids, as modeled by the Navier-Stokes equation.

\section{Countable generalized shifts and $P/poly$}\label{s:CGS}

In this section, we introduce a class of symbolic dynamical systems that contains in particular generalized shifts and which can be partially embedded in a countably piecewise linear map of the unit square. As we will see, our generalization is different from the analog shift map and is capable of simulating Turing machines with advice.

\subsection{Countable generalized shifts}\label{ss:CGS}
The broader class of symbolic systems that we introduce in this section should be thought of as a countable version of generalized shifts. Instead of changing the sequence according to a finite portion of it of fixed size (like $D_F$ and $D_G$), each sequence is changed according to a finite portion whose size depends on the sequence and can be arbitrarily large. This is a different generalization than analog shift maps, where only a finite portion of fixed size determines the image of the sequence, but infinitely many symbols can be changed in one step.\\

A countable generalized shift $\phi: A^\mathbb{Z} \rightarrow A^\mathbb{Z}$ is defined by the following information:
\begin{enumerate}
\item a set $P$ of pairs $\{(n_j,I_j)\in \mathbb{Z} \times A^{m_j}\}$, with $j\in \{0,...,N\}$ or $j\in \mathbb{N}$, such that for each $s\in A^{\mathbb{Z}}$ there is at most one $j$ such that $s_{n_j}...s_{n_j+m_j-1}=I_j$,
\item a map $J$ assigning to each element of $p=(n_j,I_j)\in P$ a word $J(p)=I_j'\in A^{m_j}$,
\item a map $H: P \rightarrow \mathbb{Z}$.
\end{enumerate}
We denote by $S_P \subset A^{\mathbb{Z}}$ the set of sequences $s\in A^\mathbb{Z}$ for which there is some $(n_j,I_j)\in P$ such that $s_{n_j}...s_{n_j+m_j-1}=I_j$.\\

The dynamical system is described as follows. Given some $s\in A^{\mathbb{Z}}$, if $s\not \in S_P$ then $\phi(s):=s$. Otherwise, let $p=(n_j,I_j)$ be the only pair assigned to $s$. The sequence $\phi(s)$ is obtained by changing the symbols in positions $n_j,...,n_j+m_j-1$ by $J(p)$, and then shifting by $H(p)$. \\

\paragraph{\textbf{Notation.}} To simplify notation, given a pair $p=(n_j,I_j) \in P$ we say that the symbols of the word $I_j$ are at positions $n_j...n_j+m_j-1$ of $I_j$ (instead of positions $1,..,m_j$). Furthermore, we say that a sequence $s\in A^{\mathbb{Z}}$ \textit{coincides} with $p$ (or with $I_j$) if
\begin{equation}\label{eq:coincP}
s_{n_j}...s_{n_j+m_j-1}=I_j.
\end{equation}
Similarly, if we denote a finite word with subindices, as in
$$w= w_{n}w_{n+1}...w_{m-1}w_{m},$$
 with $n,m\in \mathbb{Z}$, we say that $s$ \textit{coincides} with $w$ if
\begin{equation}\label{eq:coinc}
s_{n}...s_{m}=w_{n}...w_{m},
\end{equation}
where the left hand side denotes the symbols in position $n,n+1,...,m$ of $s$.\\

It is easy to see that a generalized shift is, in particular, a countable generalized shift.
\begin{lemma}\label{lem:finite}
Any generalized shift is a countable generalized shift with a finite set $P$.
\end{lemma}
\begin{proof}
Let $\phi: A^{\mathbb{Z}} \longrightarrow A^{\mathbb{Z}}$ be a generalized shift. Without loss of generality, we can assume that $D_F=D_G$, simply by taking the union of both domains and redefining $F$ and $G$ appropriately. Hence $\phi$ is defined by $F:A^l\rightarrow \mathbb{Z}$ and $G:A^l\rightarrow A^l$, where $D_F=D_G=\{i,...,i+l-1\}$. Let us define a countable generalized shift that is equal to $\phi$. As space of pairs, we choose
$P=\{ (i, (t_1...t_l)) \enspace | \enspace (t_1,...,t_l)\in A^l\}.$
Now we define $J((i, (t_1...t_l)))= G(t_1...t_l)$ and $H((i, (t_1...t_l))=F(t_1...t_l)$. We obtain a countable generalized shift $\psi$ such that $\psi(s)=\phi(s)$ for each $s\in A^\mathbb{Z}$.
\end{proof}

An interesting property of the set $S_P \subset A^{\mathbb{Z}}$ is that it can never be equal to $A^{\mathbb{Z}}$ whenever $P$ is not finite. The diagonal argument used in the proof of this lemma will be useful throughout the paper.

\begin{lemma}\label{lem:diag}
There is no countable generalized shift satisfying that $P$ is not a finite set and $S_P=A^{\mathbb{Z}}$.
\end{lemma}

\begin{proof}
Let $\phi$ be a countable generalized shift such that $P$ is not a finite set. Then there is an infinite sequence $p_{i_k}=(n_{i_k},I_{i_k})\in P$, with $I_{i_k} \in A^{m_{i_k}}$ such that $|m_{i_k}|$ or $|n_{i_k}|$ go to infinity as $k$ goes to infinity. To simplify, assume that we found a family $I_{i_k}$ such that $m_{i_k}\rightarrow \infty$, an analogous argument works for the other cases. Choose a family of sequences $s_k$ coinciding with each $p_{i_k}$. Endow $A^{\mathbb{Z}}$ with the metric
$$d(t,t')=\sum_{i=0}^k {(2N)}^{-k} (|t_k-t_k'|+|t_{-k}-t_{-k}'|),$$
where $N$ is the cardinality of $A$. Then $A^\mathbb{Z}$ is compact with this metric and the sequence $s_k$ admits a convergent subsequence $s_{k_r}$ such that $s_{k_r}\rightarrow \tilde s$ as $k_r\rightarrow \infty$. We will show that $\tilde s\not \in S_P$. Indeed, assume that there is some $\tilde p=(\tilde n, \tilde I)\in P$ (with $\tilde I$ of size $\tilde m$) such that $\tilde s$ coincides with $\tilde p$. Choose some $M$ such that $M>|\tilde n|$ and $M>|\tilde m|$. Since $s_{k_r}\rightarrow \tilde s$, there is some $K_0$ such that $|s_{k_r}-\tilde s|< N^{-M}$ for every $k_r>K_0$. This implies that $s_{k_r}$ and $\tilde s$ are equal for symbols in the positions $-M,...,M$, and hence $s_{k_r}$ coincides with $\tilde p$ for every $k_r>K_0$. We deduce that $s_{k_r}$ coincides both with $\tilde p$ and $p_{i_{k_r}}$. For a big enough $k_r$ the element $p_{i_{k_r}}$ is such that $m_{i_{k_r}}>M$, and hence $\tilde p \neq p_{i_{k_r}}$. This is a contradiction with the definition of a countable generalized shift: every sequence coincides with at most one element in $P$.
\end{proof}

It is possible to characterize, although with a property that is difficult to verify for a given example, those countable generalized shifts that are generalized shifts too.

\begin{lemma}
A countable generalized shift $\phi$ is a generalized shift if and only if there is some $N\in \mathbb{N}$ for which we can associate to each word $w=w_{-N}...w_N \in A^{2N+1}$ an integer $k_w\in \mathbb{Z}$ and a word $w'=w_{-N}'...w_N'$ satisfying the following conditions. Given $s\in A^\mathbb{Z}$, whose symbols at positions $-N,..,N$ determines a unique word $w_s$, we have:
\begin{itemize}
\item[-] If $s\in S_P$ coincides with some $p=(n_j,I_j)$, then $H(p)=k_{w_s}$ and $J(p)$ is such that for every position $r\in \{-N,...,N\}$, either $r\in \{n_j,...,n_j+m_j-1\}$ and $J(p)_r=(w_s)_r$ or $r\not \in  \{n_j,...,n_j+m_j-1\}$ and then $(w_s')_r=(w_s)_r$.
\item[-] If $s \not \in S_P$, then changing the symbols of $s$ at positions $-N,...,N$ by $w_s'$ and shifting by $k_{w_s}$ recovers $s$.
\end{itemize}
\end{lemma}

\begin{proof}
Let $\phi:A^{\mathbb{Z}} \longrightarrow A^{\mathbb{Z}}$ be a countable generalized satisfying the property in the statement. Define a generalized shift $\hat \phi:A^{\mathbb{Z}} \longrightarrow A^{\mathbb{Z}}$ by taking $D_F=D_G=\{-N,...,N\}$ and functions
$$  F(w)=k_w, \text{ for each } w\in A^{2N+1}   $$
and
$$ G(w)=w', \text{ for each } w\in A^{2N+1}. $$
If $s\in S_P$ then it follows from the first item above that $\hat \phi(s)=\phi(s)$. Furthermore if $s \not \in S_P$ then by the second item above $\hat \phi(s)=s$ and by definition of $\phi$ we have $\phi(s)=s$.\\

Conversely, let $\phi:A^{\mathbb{Z}} \longrightarrow A^{\mathbb{Z}}$ be a countable generalized shift such that there is some generalized shift $\tilde  \phi:A^{\mathbb{Z}} \longrightarrow A^{\mathbb{Z}}$, defined by functions 
$$F': A^{b-a}\rightarrow \mathbb{Z}, $$
where $D_{F'}=\{a,...,b\}$ and
$$G': A^{d-c} \rightarrow A^{d-c},$$
where $D_{G'}=\{c,...,d\}$ and such that $\phi(s)=\tilde \phi(s)$ for every $s\in A^{\mathbb{Z}}$.

 First, to simplify, we can easily construct another generalized shift $\hat \phi:A^{\mathbb{Z}} \longrightarrow A^{\mathbb{Z}}$ with associated functions $F,G$ such that $D_F,D_G=\{-N,...,N\}$ for some $N$ and $\hat \phi(s)=\tilde \phi(s)$ for every $s\in A^{\mathbb{Z}}$. To do so, let $N$ denote the greatest absolute value of the elements of $D_{F'}$ and $D_{G'}$. To define $F$, given a word $w_{-N}...w_{N} \in A^{2N+1}$, we define $F(w):=F'(w_{a}...w_{b})$. To define $G$, given a word $w_{-N}...w_N$ if $G'(w_c,...,w_d)=w'_c...w'_d$ then we define
$$ G(w_{-N}...w_N):=w_{-N}...w_{c-1}w_c'...w_a'w_{a+1}...w_N.$$
Let us now show that for this $N$ the claimed property is satisfied. Fix any word $w=w_{-N}...w_N$, and denote $F(w)$ and $G(w)$ by $k_w \in \mathbb{Z}$ and $\omega'=w'_{-N}...w_N'$ respectively. Given any sequence $s$ that coincides with $w$, since $\phi(s)=\hat \phi(s)$, the image of $s$ is obtained by replacing the symbols at positions $-N,...,N$ by $w'$ and shifting by $k_w$. If $s\not \in S_P$, then $\phi(s)=s$ and item two above is satisfied. If $s\in S_P$, then there is some $p=(n_j,I_j)$ such that $s$ coincides with $p$. Writing $I_j=s_{n_j}'...s_{n_j+m_j-1}'$, let $s'$ be the sequence
$$ s'=...0w_{-N}...w_{n_j-1}s_{n_j}'...s_{n_j+m_j-1}'w_{n_j+m_j}...w_N10...$$
which also coincides with $p$. Now the image of $s'$ can be computed either by $\phi$ or by $\hat \phi$. By looking at the $1$ placed in the last position of $s'$ which is not zero, we deduce that necessarily $H(p)=k_w$. Thus the sequence obtained from $s'$ either by changing symbols at position $-N,...,N$ by $w'$, or symbols at position $n_j,...,n_j+m_j-1$ by $J(p)$ are the same. In other words, $J(p)$ and $w'$ are the same in the positions they have in common, and $w'$ coincides with $w$ at those positions which are not in $\{n_j,...,n_j+m_j-1\}$. This is exactly the first property of the statement, and we finish the proof of the lemma.
\end{proof}

Let us give a sufficient criterion that is easy to check in practice to determine when a countable generalized shift is not a generalized shift.
\begin{defi}\label{def:constantinf}
We say that a countable generalized shift $\phi$ ``modifies at infinity" if there is a sequence of numbers $|r_{i_k}| \rightarrow \infty$ with $r_{i_k} \in \{n_{i_k},...,n_{i_k}+m_{i_k}-1\}$ for some $p_k=(n_{i_k},I_{i_k}) \in P$ such that the symbol at position $r_{i_k}$ of $I_{i_k}$ does not coincide with the symbol at position $r_{i_k}$ of $J(p_k)$.
\end{defi}

\begin{lemma}\label{lem:char}
If a countable generalized shift $\phi$ modifies at infinity, then $\phi|_{S_P}$ is not induced by a generalized shift.
\end{lemma}

\begin{proof}
Let $\phi$ be a countable generalized shift defined by $P, J$ and $H$, and assume that it modifies at infinity.

There is a sequence of numbers $|r_{i_k}| \rightarrow \infty$ with $r_{i_k} \in \{n_{i_k},...,n_{i_k}+m_{i_k}-1\}$ for some $p_k=(n_{i_k},I_{i_k}) \in P$ satisfying Definition \ref{def:constantinf}. By Lemma \ref{lem:GS}, there is always a shifted version of $\phi(s)$ that coincides with $s$ except maybe along positions in $D_G$. We will prove that this is not the case under our hypotheses. For a fixed $k$, let $s_k$ be a sequence that has zeroes everywhere except in positions $n_{i_k},...,n_{i_k}+m_{i_k}-1$ where it coincides with $I_{i_k}=a_{n_{i_k}}...a_{n_{i_k}+m_{i_k}-1}$, and also has a 1 at position $n_{i_k}+m_{i_k}$, i.e $s_k$ is of the form
$$s_k=...0a_{n_{i_k}}...a_{n_{i_k}+m_{i_k}-1}^k10... $$
The image $\phi(s_k)$ is obtained by changing the symbols in position $n_{i_k},...,n_{i_k}+m_{i_k}-1$ and shifting by $H(p_k)$.\\

 Let us show that for any integer $l$, the sequence $s_k$ and the sequence $\phi(s_k)$ shifted by $l$ do not coincide in some element in position greater or equal to $r_{i_k}$. To see this, given an arbitrary $l$ let $\phi(s_k)^l$ be the $l$-shifted sequence of $\phi(s_k)$. If $l>0$ (i.e. left shift), then the symbol in position $n_{i_k}+m_{i_k}>r_{i_k}$ of $s_k$ (which is a one) does not coincide with that of $\phi(s_k)^r$, which is a zero. If $l<0$ (right shift) then the symbol in position $n_{i_k}+m_{i_k}+r>r_{i_k}$ (which is a zero) does not coincide with the symbol of $\phi(s_k)^r$ in that position (which is a one). Finally when $r=0$, we have that the symbol in position $r_{i_k}$ does not coincide with that of $\phi(s_k)$ by hypothesis. We conclude that the symbol at some position greater or equal than $r_{i_k}$ of $s_k$ and the symbol at the same position of any shifted version of $\phi(s_k)$ are not equal. Since $r_{i_k}$ is arbitrarily large choosing an arbitrarily large $k$, we reach a contradiction with Lemma \ref{lem:GS}. We conclude that $\phi$ is not a generalized shift.
\end{proof}

Lemma \ref{lem:char} can be used to easily construct examples of countable generalized shifts that are not generalized shifts, even bijective ones. We shall call a countable generalized shift that is not a generalized shift an ``infinite generalized shift". We will refer to a countable generalized shift that is a generalized shift as a ``finite" generalized shift. 

\begin{Remark}\label{rem:01}
As for generalized shifts \cite[Lemma 1]{Mo}, any countable generalized shift is conjugate (perhaps injectively semi-conjugate depending on the cardinality of the alphabet)  to another one whose alphabet is $\Sigma=\{0,1\}$. This is done by identifying the symbols of the alphabet with large enough blocks of zeroes and ones.
\end{Remark}

\subsection{Computational power of countable generalized shifts} \label{s:CGScomp}

In \cite{Mo}, Moore showed that generalized shifts are equivalent to Turing machines, both from a dynamical and computational point of view. In this section, we analyze the computational power of countable generalized shifts and show that they can simulate Turing machines with advice. As done in previous sections, we restrict to polynomial-time Turing machines with polynomial advice, which define the complexity class $P/poly$.\\

\textbf{Technical assumptions on Turing machines with advice.} Given a Turing machine with advice $(T,a)$, we will assume that the first symbol of each advice string is always a zero, and we will only consider inputs $t_{in}$ of size $n$ such that $t_i\neq 0$ for $i=0,..,n$. It is clear that this does not restrict the computational power of the resulting polynomial-time Turing machines with polynomial advice. When the Turing machine $T$ that we consider is assumed to be reversible, we will assume as well that $q_0$ is not in the image of the transition function $\delta$. This can easily be assumed, as discussed in previous sections, see \cite[Section 6.1.2]{Mor}. The latter assumption ensures that if we take any Turing machine that is reversible, then adding advice to it keeps the global transition function injective. \\

Let us show how to simulate polynomial-time Turing machines with polynomial advice using countable generalized shifts.

\begin{theorem}\label{thm:polyCGS}
Let $(T, a)$ be a polynomial-time Turing machine with polynomial advice $a=\{a_n\}_{n\in \mathbb{N}}$. Then there is a countable generalized shift $\phi$ with some alphabet $A$ and an injective map $\varphi: \mathcal{P} \longrightarrow S_P\subset A^\mathbb{Z}$ such that $\Delta=\varphi^{-1}\circ \phi|_{S_P} \circ \varphi$. If $\Delta$ is injective, then we can assume that $\phi|_{S_P}$ is injective in all $S_P$. 
\end{theorem}

\begin{proof}
Let $(T,a)$ be a polynomial-time Turing machine with polynomial advice. Take the alphabet of the countable generalized shift to be $A=Q\cup \Sigma \cup \{d\}$, where $d$ is a symbol disjoint from $Q$ and $\Sigma$. Let $\varphi$ be the encoding function, which maps injectively the configurations of $T$ to sequences in $A^\mathbb{Z}$, defined as follows.
\begin{align}
\varphi: \mathcal{P} &\longrightarrow A^{\mathbb{Z}} \label{eq:confSeq} \\
 (q_0,t_{in}=(...0.t_0...t_n0...)) &\longmapsto  (...0d.q_0t_0...t_{n}d0...), \enspace t_i \neq 0, n \geq 0\label{eq:confSeq2}\\
 (q,(t_i)) &\longmapsto (...t_{-1}.qt_0...) \text{ otherwise.} \label{eq:confSeq3}
\end{align}
Let us define $\phi$ in terms of the space of pairs $P$ and the maps $J$ and $H$. The space of pairs $P$ is defined by $P_1\sqcup P_2$, where $P_1$ is infinite and given by
$$ P_1=\{(-1, (dq_0t_0...t_nd\underbrace{0...0}_\text{p(n)-1})) \enspace | \enspace n\in \mathbb{N} \text{ and } t_i \neq 0 \text{ for }i=0,...,n\},$$
where $p(n)$ is the polynomially-bounded function assigning to an input of size $n$ its advice of size $p(n)$. The second set of pairs $P_2$ is given by 
$$ P_2= \{(-1, (t_{-1}qt_0) \enspace | \enspace (t_{-1}qt_0) \in \Sigma \times Q \times \Sigma\}.  $$
We define the map $J$ on $P_1$ as  
$$J((-1, (dq_0t_0...t_nd\underbrace{0...0}_\text{p(n)-1}))=(0q_0t_0...t_na_1^n...a_{p(n)}^n),$$
where $a_n=a_1^n...a_{p(n)}^n$ is the advice string assigned to inputs of size $n$.
To define $J$ on $P_2$, for any $(q,t_0)$ consider $\delta(q,t_0)=(q',t',\varepsilon)$ and we define
\begin{equation*}
 J((-1, (t_{-1}qt_0))= \begin{cases} 
(t_{-1}  t' q') \text{ if } \varepsilon = +1 \\
(q't_{-1}t') \text{ if }  \varepsilon=-1\\
(t_{-1}q't') \text{ if } \varepsilon=0
 \end{cases}
 \end{equation*}
 Finally, define $H$ as $H(x)=0$ for each $x\in P_1$ and $H((-1,(t_{-1}qt_0))=\varepsilon$. Observe that a given sequence $s\in A^\mathbb{Z}$ coincides with at most one element in $P=P_1\cup P_2$, so $P$ can be used to define a countable generalized shift. The countable generalized shift $\phi$ is such that $\Delta= \varphi^{-1}\circ \phi_{S_P} \circ \varphi$. \\
 
 Let us show that if $\Delta$ is injective, then $\phi|_{S_P}$ is injective too. By contradiction, assume that the latter is not injective. Then there are two sequences $s, t \in S_P$ such that $s\neq t$ and $\phi(s)=\phi(t)$. The countable generalized shift $\phi$ changes at most a finite number of symbols of each sequence, which is then shifted by at most 1 position. Assume to simplify that both sequences are not shifted (an analogous argument works if they are shifted in any direction). Take the two unique pairs $p_1=(-1,I_1)$ and $p_2=(-1,I_2)$ such that $s,t$ coincide respectively with $p_1$ and $p_2$. Then $\phi(s)$ is equal to $s$ except maybe at those symbols in positions $D_1=\{-1,...,m_1-2\}$ and $\phi(t)$ coincides with $t$ except maybe at those symbols in position $D_2=\{-1,...,m_2-2\}$. In particular, for each $k\in D_1\setminus D_2$ we deduce that $t_k$ is equal to the symbol in position $k$ of $J(p_1)$, and for each $r\in D_2\setminus D_1$ we deduce that $s_r$ is equal to the symbol in position $r$ of $J(p_2)$. For each position $j\in D_1\cap D_2$, we must have that the symbol at position $j$ of $J(p_1)$ is equal to the symbol in position $j$ of $J(I_2)$.

We will now consider two auxiliary sequences $s', t'$, which we first define and show that they are different. We define $s'$ by
\begin{equation}
s'_i= \begin{cases}
0 \text{ if } i\not \in D_1\cup D_2,\\
s_i \text{ if } i \in D_1\\
J(p_2)_i \text{ if } i \in D_2\setminus D_1,
\end{cases}
\end{equation}
and the sequence $t'$ defined by
\begin{equation}
t'_i= \begin{cases}
0 \text{ if } i\not \in D_1\cup D_2,\\
t_i \text{ if } i \in D_2\\
J(p_1)_i \text{ if } i \in D_1\setminus D_2
\end{cases}
\end{equation} 
Our previous discussion shows that $\phi(s')=\phi(t')$. On the other hand, we know that 
\begin{itemize}
\item[-] $s$ and $t$ are equal in any position away from $D_1\cup D_2$,
\item[-] $s$ and $t$ are equal to $s'$ and $t'$ respectively in positions $D_1\cup D_2$,
\item[-] $s\neq t$,
\end{itemize}
so we deduce that we must have $s'\neq t'$.\\

Using the description of $\phi$, the fact that $\Delta$ is injective and the auxiliary sequences $s',t'$ we will reach a contradiction. Let us analyze case by case depending on $p_1$ and $p_2$. 

The first case is when $p_1,p_2\in P_2$. We have $D_1=D_2=\{-1,0,1\}$ and we deduce that $s'=...0s_{-1}.q_0s_10...$ and $t'=...0t_{-1}.q_0t_10...$. It follows that $s',t'\in \varphi(\mathcal{P})$ as per equation \eqref{eq:confSeq3}, which is a contradiction with the fact that $\Delta$ is injective. 

The second case is when $p_1 \in P_1$ and $p_2\in P_2$, then 
$$s'=...0d.q_0s_0...s_jd0...  \quad t'= ...0t_{-1}.qt_1... \text{ with } t_{-1},t_1\in \Sigma.$$
The sequence $\phi(s')$ has a $q_0$ in the zero position, while $\phi(t')$ has a $q\neq q_0$  in the zero position (since we assumed that $q_0$ is not in the image of $\delta$). We reached a contradiction with the fact that $\phi(s')=\phi(t')$.

The last case is when $p_1, p_2\in P_1$, then $p_1=(-1,I_1)$ and $p_2=(-1,I_2)$ with $I_1= (dq_0s_0...s_{m_1}d0...0)$ and $I_2=(dq_0t_0...t_{m_2}d0...0)$. Let $(a_1...a_{p(j)})$ denote the advice of the input $(...0.s_0...s_j0...)$ and $(b_1...b_{p(r)})$ denote the advice of the input $(...0.t_0...t_r0...)$. Assume that $m_1\geq m_2$, then $s'$ is of the form
 
 $$ s'=...0d.q_0s_0...s_jd\underbrace{0...0}_\text{p(j)-1}0..., $$
 and $t'$ is either of the form
  $$ t'=...0d.q_0t_0...t_rd\underbrace{0...0}_\text{p(r)-1}s_{p(r)+r+2}...s_ka_{1}...a_{p(j)}0... \quad \text{if } r+2+p(r)\leq j$$
  or of the form
 $$t'=...0d.q_0t_0...t_rd\underbrace{0...0}_\text{p(r)-1}a_{r+2+p(r)}...a_{p(j)}0... \quad \text{if } r+2+p(r)> j.$$
It follows from $\phi(s')=\phi(t')$ that $(...0.q_0s_0...s_ja_1...a_{p(j)}0...)$ is equal to 

$$(...0.q_0t_0...t_rb_1...b_{p(r)}a_{p(r+1)}...a_{p(j)}0...),$$
if $p(r)\geq j$, or to
$$(...0.q_0t_0...t_rb_1...b_{p(r)}s_{p(r)+1}...s_ja_0....a_{p(j)}0...),$$
if $p(r)<j$. If $m_1=m_2$, we must have $t'=s'$, which is a contradiction. In general, we deduce that $t_i=s_i$ for each $i=1,...,r$, and that $s_{r+1}=b_1$. However, by our technical assumptions, we know that $b_1=0$ and that $s_{r+1}\neq 0$, reaching a contradiction and finishing the proof of the theorem.
\end{proof}

The countable generalized shift constructed in the proof of Theorem \ref{thm:polyCGS} is an infinite generalized shift. Indeed, the description of $J$ on $P_1$ implies that the countable generalized shift modifies at infinity, so by Lemma \ref{lem:char} it is an infinite generalized shift. The simulation by a countable generalized shift of a polynomial-time Turing machine is done in real-time: a step of the countable generalized shift corresponds to a step of the machine. This is an advantage with respect to the simulation via analog shifts \cite{Sig}, where the first step is necessarily simulated in polynomial time with respect to the size of the input.

\subsection{Cantor set and map by blocks}

In this last section, we show that some countable generalized shifts can be, at least partially, understood as piecewise linear maps on a countable set of blocks of the square Cantor set, in the lines of \cite[Lemma 0]{Mo} for the case of a generalized shift.

From now on, we will make the simplifying assumption that a countable generalized shift is defined on the alphabet $A=\{0,1\}$, see Remark \ref{rem:01}. As done in previous sections we identify sequences in $\{0,1\}^\mathbb{Z}$ with points $C^2$ via the bijection introduced in Equation \eqref{eq:ide}.

\begin{lemma}\label{lem:bij}
Given a countable generalized shift $\phi$, there exists a piecewise linear and area-preserving map $f$ defined over a countable set of blocks into another countable set of blocks of the square Cantor set such that $\phi|_{S_P}=e^{-1}\circ f \circ e$. The following are equivalent:
\begin{itemize}
\item[-] $\phi|_{S_P}:S_P \rightarrow A^{\mathbb{Z}}$ is injective,
\item[-] the image blocks are disjoint.
\end{itemize}
\end{lemma}
\begin{proof}
Each element of $p_j=(j,I_j) \in P$ determines a block $B_j$ of the square Cantor determined by all those sequences $(s_i)$ such that $s_{n_j}...s_{n_j+m_j-1}$ coincides with $I_j$. Each block $B_j$ is first translated into the block determined by $J(p_j)$, then Baker's map is applied $H(p_j)$ times. The block $B_j$ might be cut into $r_j \leq 2^{|H(p_j)|}$ connected pieces when applying the Baker's map (or its inverse if $H(p_j)$ is negative) is applied $|H(p_j)|$ times. Let $B_j^k$ denote the preimages of those pieces. Then define the map 
$$f: \bigsqcup_{j\in \mathbb{N}} \bigsqcup_{k=0}^{r_j}B_j^k \rightarrow I^2$$
which coincides in each block $B_j^k$ with the translation and iteration of Baker's map associated with $B_j$. Clearly $f$ corresponds to $\Phi$ when applied to a point of the Cantor set. Observe that two image blocks intersect if and only if there is a point of the Cantor set in both blocks, which happens if and only if $f|_{e(S_P)}$ is not injective, which happens if and only if $\phi|_{S_P}$ is not injective.
\end{proof}
 By using a binary expansion instead of a ternary expansion in Equation \eqref{eq:ide}, the previous lemma allows us to visualize a countable generalized shift on $S_P$ as a countably piecewise linear map of the unit square.\\
 
It follows from the construction that if there is a finite word $w=(w_n,...,w_{n+m})$ such that any sequence $s\in A^\mathbb{Z}$ that coincides with $w$ is not in $S_P$, then we can as well define $f$ to be the identity in the block defined by $w$. Hence the piecewise map $f$ and the conjugacy with $\phi$ can be extended for other sequences that are not in $S_P$. The only problem, if we want a piecewise map defined on blocks, arises with those sequences that are obtained as the limit of a family of words obtained from pairs $p_j=(n_j,I_j)\in P$ with a size that tends to infinity.


\begin{thebibliography}{99}

\bibitem{Ax} H. B. Axelsen. \emph{Time complexity of tape reduction for reversible Turing machines}. International Workshop on Reversible Computation. Springer, Berlin, Heidelberg, 2011.

\bibitem{Ben} C. H. Bennett. \emph{Logical reversibility of computation}. IBM Journal of Research and Development 17.6 (1973): 525-532.


\bibitem{Bo} O. Bournez. \emph{How much can analog and hybrid systems be proved (super-) Turing}. Applied Mathematics and Computation 178.1 (2006): 58-71.

\bibitem{BC} O. Bournez, M. Cosnard. \emph{On the computational power of dynamical systems and hybrid systems}. Theor. Comput. Sci., 168-2 (1996), 417-459.

\bibitem{BGH}
O. Bournez, D.S. Gra\c{c}a, E. Hainry. \emph{Computation with perturbed dynamical systems}. J. Comput. Syst. Sci. 79 (2013) 714--724.

\bibitem{BGP} O. Bournez, D. S. Gra\c{c}a, A. Pouly. \emph{Polynomial time corresponds to solutions of polynomial ordinary differential equations of polynomial length}. Journal of the ACM (JAMC) 64.6 (2017): 1-76.

\bibitem{BP} O. Bournez, A. Pouly. \emph{A survey on analog models of computation}. Handbook of Computability and Complexity in Analysis. Springer, Cham, 2021. 173-226.

\bibitem{Br} M. S. Branicky. \emph{Universal computation and other capabilities of hybrid and continuous dynamical systems}. Theor. Comput. Sci. 138.1 (1995): 67-100.

\bibitem{BSR} M. Braverman, J. Schneider, C. Rojas. \emph{Space-bounded Church-Turing thesis and computational tractability of closed systems}. Physical review letters 115.9 (2015), 098701.

\bibitem{BY} M. Braverman, M. Yampolsky. \emph{Non-computable Julia sets}. J. Amer. Math. Soc. 19.3 (2006), 551-578.


\bibitem{ThC} R. Cardona. \emph{The geometry and topology of steady Euler flows, integrability and singular geometric structures}. PhD Thesis, Universitat Politècnica de Catalunya, 2021.

\bibitem{CMP2} R. Cardona, E. Miranda, D. Peralta-Salas. \emph{Turing universality of the incompressible Euler equations and a conjecture of Moore}. Int. Math. Res. Not. 22 (2022), 18092-18109.

\bibitem{CMP3} R. Cardona, E. Miranda, D. Peralta-Salas. \emph{Looking at Euler flows through a contact mirror: universality and undecidability}. Proceedings of the 8th European Congress of Mathematics (2023), Portoroz, Slovenia.

\bibitem{CMP4} R. Cardona, E. Miranda, D. Peralta-Salas. \emph{Computability and Beltrami fields in Euclidean space}. J. Math. Pures Appl. (9) 169 (2023), 50-81.

\bibitem{CMPP2} R. Cardona, E. Miranda, D. Peralta-Salas, F. Presas. \emph{Constructing Turing complete Euler flows in dimension 3}. Proc. Natl. Acad. Sci. 118 (2021) e2026818118.

\bibitem{EG} J. Etnyre, R. Ghrist. \emph{Contact topology and hydrodynamics I. Beltrami fields and the Seifert conjecture}. Nonlinearity 13 (2000) 441--458.


\bibitem{GCB1} D. S. Gra\c{c}a, M. L. Campagnolo, J. Buescu. \emph{Robust simulations of Turing machines with analytic maps and flows}. In: S.B. Cooper, B. Löwe, L. Torenvliet (Eds.), CiE 2005: New Computational Paradigms, in: Lecture Notes Comput. Sci., vol. 3526, Springer, 2005, pp. 169-179.

\bibitem{GCB2} D. S. Gra\c{c}a, M. L. Campagnolo, J. Buescu. \emph{Computability with polynomial differential equations}. Adv. Appl. Math. 40 (2008) 330--349.

\bibitem{GZ2} D. S. Graça, N. Zhong. \emph{Computing the exact number of periodic orbits for planar flows}. Transactions of the American Mathematical Society, 375.8 (2022), 5491-5538.

\bibitem{GZ} D. S. Gra\c{c}a, N. Zhong. \emph{Robust non-computability of dynamical systems and computability of robust dynamical systems}. Logical Methods in Computer Science (2024), 20.


\bibitem{KM99}
P. Koiran, C. Moore. \emph{Closed-form analytic maps in one and two dimensions can simulate universal Turing machines}. Theor. Comput. Sci. 210 (1999) 217--223.

\bibitem{Mo90}
C. Moore. \emph{Unpredictability and undecidability in dynamical systems}. Phys. Rev. Lett. 64 (1990) 2354--2357.

\bibitem{Mo} C. Moore. \emph{Generalized shifts: unpredictability and undecidability in dynamical systems}. Nonlinearity 4 (1991) 199--230.

\bibitem{Mor} K. Morita, \emph{Theory of reversible computing}, Springer Japan, 2017.

\bibitem{PR} M. B. Pour-El, I. Richards. \emph{The wave equation with computable initial data such that its unique solution is not computable}. Adv. Math. 39.3 (1981): 215-239.

\bibitem{RTY} J.H. Reif, J.D. Tygar, A. Yoshida. \emph{Computability and complexity of ray tracing}. Discrete Comput. Geom. 11 (1994) 265--288.

\bibitem{Sig} H. T. Siegelmann. \emph{Computation beyond the Turing limit}. Science 268-5210 (1995), p. 545-548.

\bibitem{Sig1} H. T. Siegelmann. \emph{The simple dynamics of super Turing theories}. Theor. Comput. Sci. 168.2 (1996): 461-472.

\bibitem{Sig2} H. T. Siegelmann. \emph{Neural networks and analog computation: beyond the Turing limit}. Springer Science \& Business Media, 2012.

\bibitem{S95} H.T. Siegelmann, E.D. Sontag. \emph{On the computational power of neural nets}. J. Comput. Syst. Sci. 50 (1995) 132--150.

\bibitem{T1} T. Tao. \emph{On the universality of potential well dynamics}. Dyn. PDE 14 (2017) 219-238.

\bibitem{T2} T. Tao. \emph{On the universality of the incompressible Euler equation on compact manifolds}. Discrete Cont. Dyn. Sys. A 38 (2018) 1553--1565.

\bibitem{TdL} F. Torres de Lizaur. \emph{Chaos in the incompressible Euler equation on manifolds of high dimension}. Invent. Math. 228 (2022) 687--715.

\bibitem{Tu} A. M. Turing. \emph{Systems of logic based on ordinals}. Proc. Lond. Math. Soc. 2.45 (1939), 161-228.

\bibitem{WZ}  K. Weihrauch, N. Zhong. \emph{Is wave propagation computable or can
wave computers beat the Turing machine?} Proc. Lond.
Math. Soc., 85.3 (2002), 312-332.


\end{thebibliography}
\end{document}